\documentclass[11pt,reqno]{amsart}

\usepackage[T1]{fontenc}
\usepackage{lmodern}
\usepackage{microtype}
\usepackage{mathtools}
\usepackage{amssymb}
\usepackage{enumitem}
\usepackage{xcolor}
\usepackage[colorlinks=true,linkcolor=blue!55!black,citecolor=blue!55!black,urlcolor=blue!55!black]{hyperref}
\usepackage[nameinlink,capitalize,noabbrev]{cleveref}
\hypersetup{
  pdftitle={An n\string^2 log log n Lower Bound for Permanent Circuits with Valid Division},
  pdfauthor={Jizhou Guo}
}

\numberwithin{equation}{section}
\newtheorem{theorem}{Theorem}[section]
\newtheorem{proposition}[theorem]{Proposition}
\newtheorem{lemma}[theorem]{Lemma}
\newtheorem{corollary}[theorem]{Corollary}
\theoremstyle{definition}

\theoremstyle{remark}
\newtheorem{remark}[theorem]{Remark}

\newcommand{\A}{\mathbb A}
\newcommand{\PP}{\mathbb P}
\newcommand{\Crit}{\operatorname{Crit}}
\newcommand{\permanent}{\operatorname{per}}
\newcommand{\rank}{\operatorname{rank}}
\newcommand{\im}{\operatorname{im}}

\title[Permanent Circuits with Valid Division]
{An $n^2\log\log n$ Lower Bound for Permanent Circuits with Valid Division}

\author{Jizhou Guo}
\email{mitsuha2021b@gmail.com}


\begin{document}

\begin{abstract}
We record lower bounds for permanent circuits with valid division over
characteristic zero, counting nonscalar multiplications and divisions while
additions and scalar operations are free.  Chapter~5 of OpenAI's
\emph{Ten Advances} supplies the block construction and critical-locus
estimate; these, with the parameter choice made here and the classical
bounds of Strassen and Baur--Strassen, give
\[
 \liminf_{n\to\infty}
 \frac{L_{\rm div}(\permanent_n)}{n^2\log_2\log_2n}\ge\frac1{12}.
\]
Our finite-parameter refinement for matching-minor polynomials
observes that the coefficient
vector at deletion size $h$ lies in a space of dimension at most
$\min\{\binom th,\binom t{d-h}\}$.
Twelve machine-checked declarations of the Lean development accompanying
OpenAI's manuscript on border determinantal complexity of the permanent
(September~24, 2026) imply, by a short written argument, the existence of
a geometric slice; a further written, unformalized application of the
same criterion gives
$L_{\rm div}(\permanent_n)\ge(n^2/119790)\log_2(n/44)$ for $n\ge1936$.
Assuming Proposition~7.3 of that manuscript gives the bound
$(n^2/86400)\log_2(n/48)$ for $n\ge1408$.
These are order-of-growth results; the elementary Hessian bound $n^2/2$
is larger at practical orders.
\end{abstract}

\maketitle

\section{Introduction}

For an $n\times n$ matrix $Z=(z_{ij})$, write
\[
 \permanent_n(Z)=\sum_{\sigma\in S_n}\prod_{i=1}^n z_{i,\sigma(i)}.
\]
Chapter~5 of \cite[Theorem~1.1]{OpenAITenProofs2026} proves
\[
 C(\permanent_n)\ge\frac{n^2}{144}(\log_2\log_2n-3),\qquad n\ge2^{16},
\]
where $C$ counts all arithmetic gates in division-free circuits over
$\mathbb C$.  Its proof combines a permanent block construction
\cite[Chapter~5, Lemma~6.1]{OpenAITenProofs2026} with finite-gradient-slice
lemmas, Baur--Strassen differentiation and affine B\'ezout
\cite[Chapter~5, Lemmas~4.1--4.2 and Proposition~4.3]{OpenAITenProofs2026}.
Chapter~5 numbering here refers to the PDF marked ``Updated August~6, 2026''.
The same geometric input is covered for valid rational computations by
Strassen's degree bound \cite{Strassen1973} and Baur--Strassen differentiation
\cite{BaurStrassen1983}; \Cref{prop:classical} records the required form.
Background is given in \cite[Chapters~7--8]{BurgisserClausenShokrollahi1997}.

Let $F$ have characteristic zero.  A circuit is a finite straight-line
program in the ambient rational function field $F(Z)$, with constants in $F$.
Additions, subtractions, scalar multiplications, and divisions by nonzero
scalars are free; each remaining multiplication or division has unit cost.
A division is valid when its denominator is nonzero in $F(Z)$.
Write $L_{\rm div,F}(f)$ for the minimum cost of such a program computing $f$,
and suppress $F$ when the assertion holds over every characteristic-zero
field.  All logarithms have base two.

The contributions are the following.
\begin{enumerate}[label=(\alph*),leftmargin=*]
\item The finite-module proof scheme of
\cite[Chapter~5, Proposition~5.1]{OpenAITenProofs2026} is strengthened by
the coefficient-space observation of \Cref{prop:critical-bound}.
The comparison uses the truncated parameter set already allowed by that proof.
\item Explicit block-parameter optimization gives the leading constant
$1/12$, already obtainable from Chapter~5's critical-locus bound.
The refinement improves finite parameters, not this leading constant.
\item The classical bounds cover valid division for the gradient slices
in Chapter~5, whose stated circuit result is division-free.
The count is applied to the ambient polynomial's closed gradient graph.
\item We give written circuit corollaries of the smooth slices in
\cite{OpenAIPermanent2026} and its Lean development \cite{OpenAILean2026}.
The smaller slice follows from twelve formal declarations by a written
argument; the larger one is used assuming the printed Proposition~7.3.
Neither source states an arithmetic-circuit bound.
\end{enumerate}

\begin{theorem}\label{thm:main}
Over every field of characteristic zero,
\[
 \liminf_{n\to\infty}
 \frac{L_{\rm div}(\permanent_n)}{n^2\log_2\log_2n}\ge\frac1{12}.
\]
\end{theorem}

The $n^2\log n$ results of \Cref{thm:formal-section,thm:printed-section}
imply \Cref{thm:main} for large $n$ under their respective inputs; it is kept
because its proof here is complete from classical results and the block
construction of Chapter~5, reproduced in \Cref{sec:specialization}.
The sharper liminf constants are recorded in \Cref{rem:smooth-constants};
\Cref{rem:hessian} compares the finite bounds with the elementary quadratic bound.

The exact block bound is \Cref{thm:parameters}; \Cref{sec:optimization}
gives its optimization, and \Cref{sec:smooth} gives the two smooth-slice
corollaries.
Version~1 presented the passage to valid division as a new technique;
version~2 attributes it to the classical degree and differentiation bounds.

\section{Finite slices and the classical criterion}\label{sec:preliminaries}

Dimensions are Krull dimensions, and critical points are geometric points.
For a polynomial $P$ in $m$ variables put
$\Crit(P)=V(\partial_1P,\ldots,\partial_mP)$.
Degree means the degree of the projective closure; for a reducible closed
set we sum the degrees of its irreducible components.

The next lemma is the construction of
\cite[Chapter~5, Lemma~4.2]{OpenAITenProofs2026}.
Its projective-avoidance proof works over any algebraically closed field
of characteristic zero and also includes the endpoint $k=m$.

\begin{lemma}[Finite gradient slice]\label{lem:finite-slice}
Let $K$ be algebraically closed of characteristic zero.  Suppose that $P$
is homogeneous of degree $d\ge2$ in $m$ variables and
$\dim\Crit(P)\le\Delta$.  If $1\le k\le m-\Delta$, there are linear maps
$W:K^k\hookrightarrow K^m$ and $Q:K^m\to K^k$ such that
$H(x)=Q\nabla P(Wx)$ is homogeneous of degree $d-1$ and
$H^{-1}(0)=\{0\}$.  
\end{lemma}

\begin{proof}
A general projective $(k-1)$-plane avoids $\PP(\Crit(P))$, since
$k+\Delta\le m$; choose $W$ with this property.  The restricted gradient
defines a projective morphism whose image has dimension at most $k-1$.
A general projective kernel of codimension $k$ avoids that image, giving $Q$.

\end{proof}

\begin{proposition}[Classical gradient-graph criterion]\label{prop:classical}
Let $f\in F[z_1,\ldots,z_N]$, where $F$ has characteristic zero.  Over an
extension $E/F$, let $B\in E^N$ and let
$V:E^D\to E^N$ and $R:E^N\to E^D$ be linear maps, with $D\ge1$.
Suppose that
\begin{equation}\label{eq:general-interface}
 H(x)=R\nabla f(B+Vx)
\end{equation}
has homogeneous components of a common degree $e\ge1$ and that
$H^{-1}(0)=\{0\}$ over $\overline E$.  Then
\begin{equation}\label{eq:classical-bound}
 L_{\rm div,F}(f)\ge\frac D3\log_2 e.
\end{equation}
\end{proposition}

\begin{proof}
Put $K=\overline E$ and fix a valid $F$-circuit of cost $L$.
It remains valid over $K$ at the same cost.
Baur--Strassen differentiation \cite[Theorem~1]{BaurStrassen1983}
computes the full gradient with nonscalar cost at most $3L$.
Strassen's degree bound for rational computations
\cite[p.~239, equation~(1.2)]{Strassen1973} therefore gives
\begin{equation}\label{eq:graph-degree}
 \deg\Gamma\le2^{3L},\qquad
 \Gamma=\{(z,\nabla f(z)):z\in\A_K^N\}\subseteq\A_K^{2N}.
\end{equation}
The polynomial graph $\Gamma$ is closed and isomorphic to $\A_K^N$.
The graph on the circuit's nonempty domain of definition is dense in it,
so the bound applies to this full graph.

The fiber count is the argument of
\cite[Chapter~5, Lemma~4.1]{OpenAITenProofs2026}, applied here to the closed
gradient graph of the ambient polynomial rather than to a specialized circuit.
Write $A=K[x_1,\ldots,x_D]$ and $h_i=H_i$.
The zero-fiber hypothesis makes the $h_i$ a homogeneous system of parameters,
hence a regular sequence, and
\[
 \operatorname{Hilb}_{A/(h)}(t)=\frac{(1-t^e)^D}{(1-t)^D},
 \qquad\dim_K A/(h)=e^D.
\]
Homogeneous lifts of a basis of $A/(h)$ generate $A$ over $K[h]$ by induction
on degree.  Thus $A$ is finite over $K[h]$, and dimension shows that the
$h_i$ are algebraically independent.  The resulting graded surjection from
the free $K[h]$-module on these lifts is an isomorphism by comparison of
Hilbert series, with each $h_i$ assigned degree $e$.
Consequently $H:\A_K^D\to\A_K^D$ is finite flat of degree $e^D$.
Characteristic zero makes its function-field extension separable, so a
general fiber consists of $e^D$ distinct points.

If $Vv=0$, then $H(v)=H(0)=0$, so $V$ is injective.
The finite dominant map $H$ is surjective and factors through $R$, hence
$R$ has rank $D$.  For a general $y\in K^D$, the affine linear space
\[
 A_y=\{(z,w):z\in B+\im V,\ Rw=y\}
\]
has dimension $N$ and degree one.  Its intersection with $\Gamma$ is
isomorphic to $H^{-1}(y)$, by
$x\mapsto(B+Vx,\nabla f(B+Vx))$; a left inverse of $V$ gives the inverse.
The affine B\'ezout inequality for closed sets \cite{Heintz1983} yields
\begin{equation}\label{eq:linear-section}
 e^D=\#(\Gamma\cap A_y)\le\deg\Gamma\le2^{3L}.
\end{equation}
Taking logarithms proves the assertion for every circuit of cost $L$.
\end{proof}

The zero fiber has length $e^D$ but consists of the single point $0$;
the linear-section inequality counts the distinct points of a general fiber.
For completeness, cut the projective closure successively by the linear
equations.  A component contained in a hyperplane retains its degree;
otherwise its hyperplane section has the same total degree with multiplicities.
Passing to reduced components and discarding components at infinity can only
decrease the sum.  Thus a finite affine linear section has at most the degree
of the original closed set.

\begin{corollary}[Smooth affine slice]\label{cor:smooth}
In the setting of \Cref{prop:classical}, suppose that
$G(x)=f(B+Vx)$ is homogeneous of degree $r\ge2$ in $D$ variables and that
its first partial derivatives have no nonzero common zero over
$\overline E$.  Then
\begin{equation}\label{eq:smooth-bound}
 L_{\rm div,F}(f)\ge\frac D3\log_2(r-1).
\end{equation}
For $D\ge2$, the gradient condition is equivalent to smoothness of the
projective hypersurface defined by $G$.
\end{corollary}

\begin{proof}
The chain rule gives $\nabla G(x)=V^{\mathsf T}\nabla f(B+Vx)$.
Apply \Cref{prop:classical} with $R=V^{\mathsf T}$ and $e=r-1$.
Euler's identity $rG=\sum_i x_i\partial_iG$ proves the equivalence.
\end{proof}

For a polynomial over $\mathbb Q$, a complex slice can be used over every
field $F$ of characteristic zero.  The nonzero ring
$F\otimes_{\mathbb Q}\mathbb C$ has a maximal ideal; its residue field
contains copies of both $F$ and $\mathbb C$, since nonzero elements of
either field are units in the tensor product.  An algebraic closure of the
residue field is a common overfield.  Homogeneity and the geometric
zero-fiber condition persist under base extension.  Apply the classical
criterion over this field to the original $F$-circuit.

\section{The critical locus of the matching-minor polynomial}
\label{sec:critical}

The proof follows the incidence-ring and finite-module scheme of
\cite[Chapter~5, Proposition~5.1]{OpenAITenProofs2026}.
The derivative equations use the column-deletion recurrence
\eqref{eq:derivative-recurrence} in place of that chapter's
M\"obius-inversion formulas (5.3)--(5.4), so the shared parameters are
the deletion-size-$h$ coefficient vectors rather than the power sums $p_B$.
The additional step is that the deletion-size-$h$ coefficient vector
lies in a space of dimension at most $\min\{\binom th,\binom t{d-h}\}$.
Work over an algebraically closed field $K$ of characteristic zero.
For a $t\times s$ matrix $X=(x_{ia})$ and $0\le d\le\min\{t,s\}$, define
\begin{equation}
 M_{t,s,d}(X)
 =\sum_{\substack{I\subseteq[t],\ J\subseteq[s]\\|I|=|J|=d}}
   \permanent(X_{I,J}).
 \label{eq:matching-minor}
\end{equation}
Equivalently, this is the weight enumerator of size-$d$ matchings in the
complete bipartite graph $K_{t,s}$.  Write $M_\ell(Y)$ for the size-$\ell$
matching sum, with $M_0(Y)=1$ and $M_\ell(Y)=0$ for $\ell<0$.

Fix $3\le d\le\min\{t,s\}$ and put $q=d-1$.  For $S\subseteq[t]$, let
$X_{\widehat S}$ denote deletion of the rows in $S$; a hat over a column has
the analogous meaning.

\subsection{A column-deletion recurrence}

Fix a column $a$ and write $u_i=x_{ia}$.  Set
\[
 F_{\ell,S}=M_\ell(X_{\widehat S}),\qquad
 G_{\ell,S}=M_\ell(X_{\widehat S,\widehat a}).
\]
Splitting a matching according to whether it uses column $a$ gives
\begin{equation}
 F_{\ell,S}=G_{\ell,S}
 +\sum_{j\notin S}u_jG_{\ell-1,S\cup\{j\}}.
 \label{eq:deletion-basic}
\end{equation}
Let $\mathcal U$ denote the second operator on the right.  It is nilpotent on the
finite degree range, and ordered choices of $r$ distinct rows give
\begin{equation}
 (\mathcal U^rF)_{\ell,S}
 =r!\!\sum_{\substack{J\subseteq[t]\setminus S\\|J|=r}}
 u_JF_{\ell-r,S\cup J},\qquad u_J:=\prod_{j\in J}u_j.
 \label{eq:deletion-power}
\end{equation}
Inverting $I+\mathcal U$ in \eqref{eq:deletion-basic} and using
$\partial_{ia}M_{t,s,d}=G_{q,\{i\}}$ yields
\begin{equation}
 \frac{\partial M_{t,s,d}}{\partial x_{ia}}
 =\sum_{r=0}^{q}(-1)^rr!
 \sum_{\substack{J\subseteq[t]\setminus\{i\}\\|J|=r}}
 u_JM_{q-r}(X_{\widehat{\{i\}\cup J}}).
 \label{eq:derivative-recurrence}
\end{equation}
The final term is $(-1)^qq!e_q(u_1,\ldots,\widehat u_i,\ldots,u_t)$.

\subsection{The shared coefficient spaces}

At deletion size $1\le h<d$, the coefficients in
\eqref{eq:derivative-recurrence} have the form
\begin{equation}
 A_S=M_{d-h}(X_{\widehat S}),\qquad |S|=h.
 \label{eq:deletion-vector}
\end{equation}
For each $(d-h)$-subset $I$ of rows, let $w_I$ be the sum of the weights of
all size-$(d-h)$ matchings using precisely the row set $I$.  Then
\begin{equation}
 A_S=\sum_{I\cap S=\varnothing}w_I.
 \label{eq:disjointness-map}
\end{equation}
Define the disjointness matrix
\[
 D_{h,d-h}=(\mathbf1[S\cap I=\varnothing])_{|S|=h,\,|I|=d-h}.
\]
Thus $(A_S)_{|S|=h}=D_{h,d-h}(w_I)_{|I|=d-h}$ lies in its column space.
This matrix has $\binom th$ rows and
$\binom t{d-h}$ columns, so
\begin{equation}\label{eq:incidence-rank}
 \rank D_{h,d-h}\le\min\left\{\binom th,\binom t{d-h}\right\}.
\end{equation}
Only this elementary rank inequality is needed.

\begin{proposition}[Critical-locus bound]\label{prop:critical-bound}
For $3\le d\le\min\{t,s\}$,
\begin{equation}
 \dim\Crit(M_{t,s,d})\le R_{t,d}+s(d-2),
 \quad
 R_{t,d}=\sum_{h=1}^{d-1}
 \min\left\{\binom th,\binom t{d-h}\right\}.
 \label{eq:critical-bound}
\end{equation}
\end{proposition}

\begin{proof}
For $1\le h<d$, choose $r_h=\rank D_{h,d-h}$ coordinates
$z_{h,1},\ldots,z_{h,r_h}$ on an abstract copy of $\im D_{h,d-h}$, and use
them to reconstruct the vector in \eqref{eq:deletion-vector}.  Substitute
these reconstructions in \eqref{eq:derivative-recurrence}, retaining
$M_0=1$, to obtain equations $Q_{ia}(u^{(a)};z)$.

Let
\begin{equation}
 \widetilde A=
 K[z_{h,j},u_i^{(a)}]/(Q_{ia}:i\in[t],a\in[s]).
 \label{eq:incidence-ring}
\end{equation}
Substituting the actual deletion vectors and matrix columns defines a
surjection from $\widetilde A$ onto the coordinate ring of
$\Crit(M_{t,s,d})$.

Introduce the parameter ring
\begin{equation}
 B=K[z_{h,j}]
 [y_{a,j}:a\in[s],\ 1\le j<q]
 \label{eq:parameter-ring}
\end{equation}
and map $y_{a,j}$ to $e_j(u^{(a)})$.  Filter the presentation of
$\widetilde A$ over $B$ by total degree in the $u$ variables, with $B$ in
degree zero.  Equation \eqref{eq:derivative-recurrence} expresses every
$e_q(u^{(a)}\setminus u_i^{(a)})$ in filtered degree at most $q-1$.

For one column $u=(u_1,\ldots,u_t)$, use
\begin{equation}
 e_q(u\setminus u_i)=\sum_{j=0}^{q}(-u_i)^{q-j}e_j(u),
 \qquad
 \sum_{i=1}^{t}e_q(u\setminus u_i)=(t-q)e_q(u).
 \label{eq:symmetric-identities}
\end{equation}
Since $t-q\ne0$, the second identity first reduces $e_q(u)$ to filtered
degree at most $q-1$.  The first then reduces each $u_i^q$ to smaller filtered
degree.  Repeated reduction shows that $\widetilde A$ is spanned over $B$ by
the finitely many monomials
\[
 \prod_{i,a}(u_i^{(a)})^{\nu_{ia}},\qquad 0\le\nu_{ia}<q.
\]
Hence $\widetilde A$ is finite over $B$.  It follows from
\eqref{eq:incidence-rank} that
\[
 \dim\Crit(M_{t,s,d})
 \le\dim\widetilde A\le\dim B
 =\sum_{h=1}^{d-1}r_h+s(q-1)
 \le R_{t,d}+s(q-1),
\]
which is \eqref{eq:critical-bound}.
\end{proof}

\begin{corollary}[Disjoint blocks]\label{cor:block-critical}
Let $X^{(1)},\ldots,X^{(b)}$ be disjoint $t\times s$ variable blocks and
$\lambda_1,\ldots,\lambda_b\in K^\times$.  Then
\[
 \dim\Crit\!\left(\sum_{h=1}^{b}\lambda_hM_{t,s,d}(X^{(h)})\right)
 \le b\bigl(R_{t,d}+s(d-2)\bigr).
\]
\end{corollary}

\begin{proof}
As in \cite[Chapter~5, Corollary~5.2]{OpenAITenProofs2026}, the nonzero
weights are units, so the critical locus is the product of the critical loci
of the summands.  Apply \Cref{prop:critical-bound} to each factor and add dimensions.
\end{proof}

\section{A root-of-unity specialization of the permanent}
\label{sec:specialization}

We reproduce the block construction of \cite[Chapter~5, Lemma~6.1]{OpenAITenProofs2026} in the
notation needed below.  Let $b,t,s\ge1$ and $1\le d\le\min\{t,s\}$.
Partition $[r]$ into blocks
$R_1,\ldots,R_b$ of size $t$, where $r=bt$, and choose a primitive $d$th root
of unity $\zeta$.  Put $\theta=(-1)^{d+1}2^d$.  Form an
$r\times(r-d)$ constant matrix $U$ from the following columns:
\begin{enumerate}[label=(\roman*)]
\item for each $h\in[b]$, take $t-d$ copies of the indicator vector
      $\mathbf1_{R_h}$;
\item for $2\le h\le b$ and $0\le j<d$, take
\[
 \mathbf1_{R_1\cup\cdots\cup R_{h-1}}+2\zeta^j\mathbf1_{R_h}.
\]
\end{enumerate}
For an $r\times s$ variable matrix $X$, define the square matrix
\begin{equation}
 \mathcal B(X)=
 \begin{pmatrix}X&U\\ \mathbf1&0\end{pmatrix},
 \label{eq:block-matrix}
\end{equation}
where the lower-left all-ones block has $s-d$ rows.  Thus $\mathcal B(X)$ has
size $r+s-d$.

\begin{proposition}[Block specialization]\label{prop:block-specialization}
There are nonzero $\lambda_1,\ldots,\lambda_b\in K$ such that
\begin{equation}
 \permanent(\mathcal B(X))
 =c\sum_{h=1}^{b}\lambda_hM_{t,s,d}(X_{R_h,[s]}),
 \quad
 c=(s-d)!(t-d)!(t!)^{b-1}\ne0.
 \label{eq:block-specialization}
\end{equation}
\end{proposition}

\begin{proof}
Work in the square-zero algebra
$K[z_1,\ldots,z_r]/(z_1^2,\ldots,z_r^2)$ and put
$y_h=\sum_{i\in R_h}z_i$ and $Y_h=y_1+\cdots+y_h$.  The product of the column
forms of $U$ is
\begin{equation}
 \left(\prod_{h=1}^{b}y_h^{t-d}\right)
 \prod_{h=2}^{b}(Y_{h-1}^d+\theta y_h^d),
 \label{eq:column-form-product}
\end{equation}
because
\[
 \prod_{j=0}^{d-1}(Y+2\zeta^jy)=Y^d+\theta y^d.
\]
The square-zero relations imply inductively that \eqref{eq:column-form-product}
equals
\begin{equation}
 \sum_{h=1}^{b}\lambda_hy_h^{t-d}\prod_{g\ne h}y_g^t,
 \label{eq:unsaturated-block}
\end{equation}
where
\[
 \lambda_1=\theta^{b-1},\qquad
 \lambda_h=\theta^{b-h}(1+\theta)^{h-2}\quad(2\le h\le b).
\]
All these coefficients are nonzero.

For any $r\times q$ matrix $A$ and $|I|=q$, the coefficient of $z_I$ in the
product of its column forms is $\permanent(A_{I,[q]})$.  Since
$y_h^j=j!\sum_{I\subseteq R_h,|I|=j}z_I$, equation
\eqref{eq:unsaturated-block} says that an $(r-d)\times(r-d)$ minor of $U$
vanishes unless the omitted $d$ rows lie in one block $R_h$; in that case it
equals $\lambda_h(t-d)!(t!)^{b-1}$.

Expanding the permanent of \eqref{eq:block-matrix} first along its lower rows
gives
\[
 \permanent(\mathcal B(X))=(s-d)!
 \sum_{\substack{I\subseteq[r],J\subseteq[s]\\|I|=|J|=d}}
 \permanent(U_{I^c,[r-d]})\permanent(X_{I,J}).
\]
The preceding minor identity leaves exactly the terms in
\eqref{eq:block-specialization}.
\end{proof}

\begin{theorem}[Block parameter bound]\label{thm:parameters}
Let $b,t,s,d$ be integers with $b\ge1$ and $3\le d\le\min\{t,s\}$.  Set
\[
 n=bt+s-d,\quad m=bts,\quad
 \Delta=b\bigl(R_{t,d}+s(d-2)\bigr),\quad k=m-\Delta.
\]
If $k>0$, then, over every field of characteristic zero,
\begin{equation}\label{eq:parameter-bound}
 L_{\rm div}(\permanent_n)\ge\frac k3\log_2(d-1).
\end{equation}
\end{theorem}

\begin{proof}
Work over an algebraic closure $K$ of the circuit's ground field.
\Cref{prop:block-specialization} gives an affine map $X\mapsto B+T X$
whose permanent is the degree-$d$ block polynomial $P$ in $m=bts$ variables.
\Cref{cor:block-critical} gives $\dim\Crit(P)\le\Delta$.
By \Cref{lem:finite-slice}, choose $W:K^k\hookrightarrow K^m$ and
$Q:K^m\to K^k$ for which $Q\nabla P(Wx)$ has only the origin over zero.
The chain rule gives
\[
 Q\nabla P(Wx)=QT^{\mathsf T}\nabla\permanent_n(B+TWx).
\]
Use \Cref{prop:classical} with $V=TW$, $R=QT^{\mathsf T}$, and $e=d-1$.
\end{proof}

\section{Optimization and the leading constant}\label{sec:optimization}

\begin{proof}[Proof of \Cref{thm:main}]
For sufficiently large $n$, put $\ell=\log_2n$ and $\lambda=\log_2\ell$, and choose the following parameters; $b$ and $s$ follow
\cite[Chapter~5, equation~(7.1)]{OpenAITenProofs2026}:
\begin{equation}\label{eq:optimized-parameters}
 t=\lfloor\ell-4\lambda\rfloor,\quad
 d=\lfloor t/\lambda\rfloor,\quad
 b=\lfloor n/(2t)\rfloor,\quad s=n-bt+d.
\end{equation}
The parameter conditions hold eventually, $bt\in[n/2-t,n/2]$, and $s\ge n/2$.
Consequently
\begin{equation}\label{eq:m-lower}
 m=bts\ge\frac{n^2}{4}-\frac{nt}{2},\qquad
 m=(1/4+o(1))n^2.
\end{equation}
Since $R_{t,d}\le2^t-1$ and $2^t\le n/\ell^4$,
\begin{equation}\label{eq:delta-ratio}
 \frac\Delta m
 =\frac{R_{t,d}}{ts}+\frac{d-2}{t}
 \le\frac{2}{t\ell^4}+\frac1\lambda=o(1).
\end{equation}
Thus $k=(1/4-o(1))n^2$.  Also $t/\ell\to1$ and
$d\sim\ell/\lambda$, so
\begin{equation}\label{eq:degree-asymptotic}
 \log_2(d-1)=\lambda-\log_2\lambda+o(1)=(1-o(1))\lambda.
\end{equation}
\Cref{thm:parameters} gives
$L_{\rm div}(\permanent_n)\ge(1/12-o(1))n^2\log_2\log_2n$.
\end{proof}

Chapter~5's full parameter bound
$\Delta_{\rm old}=b(2^t-1+s(d-2))$ gives the same leading constant with
\eqref{eq:optimized-parameters}: \eqref{eq:delta-ratio} used only
$R_{t,d}\le2^t-1$, so it holds verbatim for $\Delta_{\rm old}/m$.
A fair finite comparison truncates its
power-sum parameters: the partitions in
\cite[Chapter~5, equation~(5.3)]{OpenAITenProofs2026} use only nonempty
$B$ with $|B|<d$.  Put
\[
 T_{t,d}=\sum_{h=1}^{d-1}\binom th.
\]
That proof therefore already gives $b(T_{t,d}+s(d-2))$.
For the same parameters, the refinement improves the surviving dimension by
\begin{equation}\label{eq:finite-improvement}
 k-k_{\rm trunc}=b(T_{t,d}-R_{t,d}),\qquad
 R_{t,d}\le T_{t,d}\le2^t-1.
\end{equation}

\begin{remark}[Finite optimization]\label{rem:finite-table}
Maximizing the block bound over integer $b,t,s,d$ with $n=bt+s-d$ gives
the following values divided by $n^2$, rounded to six decimal places:
\begin{center}
\begin{tabular}{c|ccc}
$n$ & full $2^t-1$ & truncated $T_{t,d}$ & refined $R_{t,d}$\\\hline
$2^{16}$ & 0.131411 & 0.148971 & 0.189746\\
$2^{32}$ & 0.189340 & 0.215318 & 0.263163
\end{tabular}
\end{center}
For fixed $t,d$, write $a=t-d+2$ and let $S$ be the coefficient budget.
Then $k(b)=-atb^2+(a(n+d)-S)b$; its integer maximum is at a clipped
floor or ceiling of the rational vertex.
The optimizing $(t,d)$ are $(14,6)$, $(22,6)$ and $(63,7)$ in the first row,
and $(30,10)$, $(57,9)$ and $(111,13)$ in the second.
Positive bounds require $d\le4\log_2n+1$ and
$\binom t{\lfloor d/2\rfloor}<tn$, giving finite enumeration ranges.
The values are obtained with exact integer arithmetic and rational
upper and lower bounds for the logarithms.
At $n=2^{16}$ the published Chapter~5 baseline is $n^2/144$.
For $t\ge d$, we have $\binom th\le\binom t{d-h}$ when $h\le d-h$,
so
\[
 R_{t,d}=\sum_{h=1}^{d-1}\binom t{\min\{h,d-h\}},\qquad
 T_{t,\lfloor d/2\rfloor+1}\le R_{t,d}\le2T_{t,\lfloor d/2\rfloor+1}.
\]
For fixed $t,s$, if $d\ge4$ satisfies $s(t-d+2)>R_{t,d}$, then
$d'=\lfloor d/2\rfloor+1$ satisfies $s(t-d'+2)>T_{t,d'}$.
Conversely, $R_{t,d}\le T_{t,d}$ makes every truncated-feasible depth
refined-feasible.  If the largest refined-feasible depth is at least four, it is less than
twice the largest truncated depth, and $\log_2(d-1)$ changes by at most one.
\end{remark}

\begin{remark}[Limit of the block method]\label{rem:block-limit}
The leading constant $1/12$ is the limit of these block bounds.
Indeed, positivity means $R_{t,d}<(t-d+2)s$.
Since $b\ge1$ and $t\ge d$, we have $s=n-bt+d\le n$ and
$(t-d+2)s\le tn$; also $s\ge d$ gives $t\le n$.
Thus $R_{t,d}<tn\le n^2$.
With $h=\lfloor d/2\rfloor$ and the formula in \Cref{rem:finite-table},
$R_{t,d}\ge\binom th\ge\binom dh\ge2^d/(d+1)$.
It follows first that $d=O(\log n)$ and then
$d\le2\log_2n+O(\log\log n)$.
Writing $r_0=bt$, we have $r_0+s=n+d$ and
$m=r_0s\le(n+d)^2/4=(1/4+o(1))n^2$.
Since $k\le m$ and $\log_2(d-1)\le\log_2\log_2n+O(1)$,
the block bound has leading coefficient at most $1/12$.
The proof of \Cref{thm:main} attains this coefficient already with the
untruncated Chapter~5 estimate.
\end{remark}

\section{Circuit corollaries of smooth permanent slices}\label{sec:smooth}

\subsection{The slice from the formal statements}

For an integer $\kappa\ge128$, put
\begin{equation}\label{eq:formal-parameters}
 q=\left\lfloor\frac{\kappa-1}{3}\right\rfloor,\qquad
 r=\kappa-2q,\qquad
 D=\left\lfloor\frac{\binom{q+2}{2}}{18}\right\rfloor.
\end{equation}
The size $D$ belongs to the accompanying Lean development
\cite{OpenAILean2026}; it does not appear in the printed manuscript
\cite{OpenAIPermanent2026}, whose Proposition~7.3 uses a larger size.
We use the development at commit
\begin{center}\small
\href{https://github.com/openai/math/tree/adc7f1241b42e322a6451854ab7e4b4c146bf78a}{\texttt{adc7f1241b42e322a6451854ab7e4b4c146bf78a}}.
\end{center}
The frozen sources were rebuilt locally with Lean~4.34.1 and Mathlib
\begin{center}\small\texttt{d13f23b723b8a846827a245b89c10fc7d3f11612}.
\end{center}
For each of the twelve declarations, a separate file restates the expected
type and closes it with \texttt{exact}, and \texttt{\#print axioms} reports
only \texttt{propext}, \texttt{Classical.choice} and \texttt{Quot.sound}.
The derivation of the slice property from these declarations (the proof
of \Cref{prop:formal-slice}) and the deduction of the circuit bound from
it are written mathematics and are not formalized here.

All names below have prefix \path{OAI.PermanentBorder.}:
\begin{quote}\small\noindent
\path{exists_effective_critical_bound};\par
\path{EffectiveParameters.sectionSize_bounds};\par
\path{Selector.integralCoefficientFamily_homogeneous};\par
\path{Selector.integralCoefficientFamily_specialize};\par
\path{Selector.integralCoefficientFamily_complex};\par
\path{SliceParameters.nonzero_of_gradientDimension};\par
\path{SliceParameters.exists_complex_smooth_section_finite};\par
\path{SliceParameters.linearSection_eq_affineSubstitution};\par
\path{selectorScalar_ne_zero};\par
\path{exists_selector_affine_projection};\par
\path{Parameters.index_bounds};\par
\path{affineSubstitution_permanent}.
\end{quote}

\begin{proposition}[Slice property]\label{prop:formal-slice}
The twelve declarations above imply that, for every $\kappa\ge128$ and
$n\ge11\kappa-2q$, there is an affine-linear map
$\Phi_n:\mathbb C^D\to\operatorname{Mat}_n(\mathbb C)$ such that
$G=\permanent_n\circ\Phi_n$ is nonzero, homogeneous of degree $r$, and
$\Crit(G)=\{0\}$.
\end{proposition}

\begin{proof}
The parameter declarations give
$2\le D\le\kappa^2$, $D-1\ge\kappa^2/330$, $q\le\kappa/2$,
$r\ge2$, and $r-1\ge\kappa/4$.
The integral coefficient family is
\[
 F_\kappa(\theta;X)
 =[u^qv^q]\permanent_\kappa(X+uI+v\operatorname{diag}(\theta)),
\]
homogeneous of degree $r$ in its $\kappa^2$ matrix variables.
Over the infinite perfect field $\overline{\mathbb F}_2$,
\path{exists_effective_critical_bound} supplies parameters for which the
gradient quotient has dimension at most $\kappa^2-D$.
The integral-family specialization identity identifies this polynomial
with the specialization of $F_\kappa$; the dimension bound makes it nonzero.
Apply \path{exists_complex_smooth_section_finite} with ambient cardinality
$\kappa^2$, dimension bound $\kappa^2-D$, and section binder $D-1$.
Its strict dimension condition is
$(\kappa^2-D)+(D-1)=\kappa^2-1<\kappa^2$.
It gives complex parameters and a linear section $g$ in exactly $D$ variables,
homogeneous of degree $r$, for which
\[
 g\ne0,\qquad g(x)=0\ \text{and}\ \nabla g(x)=0\ \Longrightarrow\ x=0.
\]
Euler's identity over $\mathbb C$ gives $\Crit(g)=\{0\}$.
The complex-family identity identifies the specialization with
the coefficient polynomial in the selector declaration.
Write $A(X)=b+BX$ for its affine-linear matrix.
The declaration \path{affineSubstitution_permanent} identifies the left side.
Thus the projection identity is
\begin{align*}
 \permanent_n(A(X))&=a_\kappa F_\kappa(\theta';X),\\
 a_\kappa&=(-2)^{2\kappa}((\kappa-q)!)^4\ne0.
\end{align*}
The linear section has coefficient matrix $H$.
The section-to-substitution identity gives
\[
 g(x)=F_\kappa(\theta';Hx).
\]
Substitution $X=Hx$ is a polynomial-ring homomorphism.  Composing gives
\[
 \Phi_n(x)=b+BHx,\qquad \permanent_n(\Phi_n(x))=a_\kappa g(x).
\]
Take $G=a_\kappa g$.  The nonzero scalar preserves the degree and critical locus.
The projection requires $q\le\kappa$ and $n\ge11\kappa-2q$, as assumed.

\end{proof}

\begin{theorem}[Circuit bound from the formal slice]\label{thm:formal-section}
Let $\kappa\ge128$ be an integer and let $q,r,D$ be as in
\eqref{eq:formal-parameters}.  The slice property of \Cref{prop:formal-slice}
gives, over every field of characteristic zero and for every
$n\ge11\kappa-2q$,
\begin{equation}\label{eq:formal-exact}
 L_{\rm div}(\permanent_n)\ge\frac D3\log_2(r-1).
\end{equation}
In particular, for $n\ge1936$,
\begin{equation}\label{eq:formal-finite}
 L_{\rm div}(\permanent_n)\ge\frac{n^2}{119790}\log_2(n/44)
 \ge\frac{n^2\log_2n}{239580}.
\end{equation}
\end{theorem}

\begin{proof}
Apply \Cref{cor:smooth} over a common algebraically closed overfield of
the circuit's field and $\mathbb C$, as in \Cref{sec:preliminaries}.
For $n\ge1936$, take $\kappa=\lfloor3(n-2)/31\rfloor$.
Since $q\ge(\kappa-3)/3$,
\[
 11\kappa-2q\le31\kappa/3+2\le n,\qquad
 \kappa\ge(3n-37)/31\ge n/11>128.
\]
Hence $D\ge n^2/(330\cdot121)$ and $r-1\ge n/44$.
Equation~\eqref{eq:formal-exact} gives the first bound of
\eqref{eq:formal-finite}, since $3\cdot330\cdot121=119790$.
The second follows from $n\ge44^2$.
\end{proof}

\subsection{The larger slice in the printed proposition}

Proposition~7.3 of \cite[pp.~25--26]{OpenAIPermanent2026} gives,
for $\kappa\ge128$, a nonzero degree-$r$ homogeneous form in
$\nu=\lfloor\kappa^2/100\rfloor$ variables with smooth projective zero set,
and for every $n\ge11\kappa-2q$ an affine map $\Phi_n$ representing it as
$\permanent_n\circ\Phi_n$.
The manuscript defines smoothness by the condition that
``the projective zero set has no point at which all first partial derivatives
vanish'' \cite[Section~2]{OpenAIPermanent2026}.
Euler's identity gives critical locus $\{0\}$.
Here $\nu\sim\kappa^2/100$ is larger than $D\sim\kappa^2/324$;
the larger size is not supplied by the twelve formal declarations.
The repository's \texttt{README.md} notes that unformalized results may
contain issues \cite{OpenAILean2026}; this printed proposition is used only
as an assumption.

\begin{theorem}[Conditional bound]\label{thm:printed-section}
Assume Proposition~7.3 of \cite{OpenAIPermanent2026} as stated above.
Over every field of characteristic zero, for $\kappa\ge128$ and
$n\ge11\kappa-2q$,
\begin{equation}\label{eq:printed-exact}
 L_{\rm div}(\permanent_n)\ge\frac\nu3\log_2(r-1),
 \qquad \nu=\lfloor\kappa^2/100\rfloor.
\end{equation}
For every $n\ge1408$, this implies
\begin{equation}\label{eq:printed-finite}
 L_{\rm div}(\permanent_n)\ge\frac{n^2}{86400}\log_2(n/48).
\end{equation}
\end{theorem}

\begin{proof}
Apply \Cref{cor:smooth} to the stated slice over a common algebraically
closed overfield of $F$ and $\mathbb C$.
For $\kappa\ge128$,
$r-1\ge(\kappa-1)/3\ge\kappa/4$ and
$\nu\ge\kappa^2/100-1\ge\kappa^2/200$.
Take $\kappa=\lfloor n/11\rfloor$.
For $n\ge1408$, we have $\kappa\ge128$,
$11\kappa-2q\le n$, and $\kappa\ge n/11-1\ge n/12$.
Thus $\nu/3\ge n^2/(3\cdot200\cdot144)=n^2/86400$ and
$r-1\ge n/48$, proving \eqref{eq:printed-finite}.
\end{proof}

\begin{remark}\label{rem:smooth-constants}
The liminf of $L_{\rm div}(\permanent_n)/(n^2\log_2n)$ is at least
$1/103788$ from the formal statements, and at least $3/96100$ assuming
Proposition~7.3.  Use $\kappa=\lfloor3(n-2)/31\rfloor$, so
$\kappa/n\to3/31$, $D/\kappa^2\to1/324$, $\nu/\kappa^2\to1/100$,
and $\log_2(r-1)/\log_2n\to1$.
The products are $(1/3)(1/324)(3/31)^2=1/103788$ and
$(1/3)(1/100)(3/31)^2=3/96100$.
\end{remark}

\begin{remark}[Hessian bound]\label{rem:hessian}
In this cost model, $L_{\rm div}(\permanent_n)\ge n^2/2$ for $n\ge3$,
by the Hessian-rank method that Mignon and Ressayre
\cite{MignonRessayre2004} use for determinantal complexity.
Write $H_u$ for the Hessian of $u$ at a point where the circuit is defined.
A multiplication $uv$ has inherited terms $vH_u+uH_v$ and the new term
\[
 (\nabla u)(\nabla v)^{\mathsf T}+(\nabla v)(\nabla u)^{\mathsf T}.
\]
A division $u/v$ has inherited terms $v^{-1}H_u-uv^{-2}H_v$ and the new term
\[
 -v^{-2}\bigl((\nabla u)(\nabla v)^{\mathsf T}
 +(\nabla v)(\nabla u)^{\mathsf T}\bigr)
 +2uv^{-3}(\nabla v)(\nabla v)^{\mathsf T}.
\]
Each new term has rank at most two.  Induction through the free linear
operations expresses the output Hessian as a linear combination of one
such term per nonscalar gate, so its rank is at most $2L$.
At the all-ones matrix, $\partial^2\permanent_n/\partial z_{ij}\partial z_{kl}$
equals $(n-2)!$ if $i\ne k$ and $j\ne l$, and zero otherwise.  Hence
\[
 \operatorname{Hess}\permanent_n=(n-2)!(J_n-I_n)\otimes(J_n-I_n).
\]
Here $J_n$ is the all-ones matrix.  Since
$(J_n-I_n)^{-1}=J_n/(n-1)-I_n$ for $n\ge2$, this Hessian has rank $n^2$.
Its determinant is a nonzero polynomial in $Z$, so the full-rank locus is
a nonempty open set.  Over an algebraic closure $K$, the circuit's domain
of definition consists of points of $K^{n^2}$ where every gate is regular
and every divisor is nonzero.  It is open and nonempty because the divisors
are valid, and meets the full-rank locus.  The rank bound at a point in
their intersection proves the assertion even if a divisor vanishes at
the all-ones matrix.

Every explicit finite bound here is smaller than $n^2/2$ at practical
orders; their content is the order of growth.  The first bound in
\eqref{eq:formal-finite} exceeds $n^2/2$ exactly when
$n>44\cdot2^{59895}$, and its second bound does so exactly when
$n>2^{119790}$.  Equation~\eqref{eq:printed-finite} requires
$n>48\cdot2^{43200}$.  The numerical Chapter~5 baseline exceeds $n^2/2$
only for $n>2^N$ with $N=2^{75}$ (and is stated for the different
all-gate measure $C$).
By a direct computation in exact arithmetic, the optimized block bounds
first exceed $n^2/2$ on the sequence $n=2^e$ at $e=881$, $453$ and $267$
for the full, truncated and refined coefficient budgets, respectively.
\end{remark}

\section*{Use of AI}
The proofs and the text of this paper were produced with AI systems (GPT-5.6 Sol, GPT-6.1 Sol and Claude) under the author's direction, and were checked by independent reviews with AI models.
The author takes responsibility for the content.

\bibliographystyle{amsplain}
\bibliography{references}

\end{document}